\documentclass[a4paper,11pt]{article}

\usepackage[margin=2.50cm]{geometry}
\usepackage{float}
\usepackage{graphicx}
\usepackage[unicode]{hyperref}
\usepackage{amsthm}
\usepackage{amssymb}
\usepackage{amsmath}
\usepackage{mathrsfs}
\usepackage{verbatim}
\usepackage{mathtools}
\usepackage{xspace}
\usepackage{color}
\usepackage{url}
\usepackage{xspace}
\usepackage[all]{xy}
\usepackage{amsmath}
\usepackage{mathtools}
\usepackage{verbatim}

\newcommand{\gdist}{d} 
\newcommand{\dist}[2]{\gdist(#1,#2)}

\DeclareMathOperator{\aff}{aff} 

\newcommand{\del}{{\rm Del}} 
\newcommand{\vor}{{\rm Vor}}

\newcommand{\distance}[2]{\|#1-#2\|}

\newcommand{\R}{\mathbb{R}}
\newcommand{\Z}{\mathbb{Z}}
\newcommand{\dotp}[2]{\vec{#1}\cdot\vec{#2}}

\newcommand{\ignore}[1]{}
\newcommand{\Pro}{\mathcal{P}}
\newcommand{\str}{{\rm star}}
\newtheorem{lemma}{Lemma}
\newtheorem{defi}[lemma]{Definition}

\newtheorem{theorem}[lemma]{Theorem}

\newtheorem{corollary}[lemma]{Corollary}

\newtheorem{remark}[lemma]{Remark}

\begin{document}

\title{Delaunay simplices in diagonally distorted lattices}
\author{
Aruni Choudhary\footnote{Freie Universit\"at Berlin, Berlin, Germany}
\and 
Arijit Ghosh\footnote{The Institute of Mathematical Sciences, Chennai, India}
}
\maketitle

\begin{abstract}
\label{abstract}
\emph{Delaunay protection} is a measure of how far is a Delaunay triangulation 
from being degenerate. 
In this short paper we study the protection properties and other 
quality measures of the Delaunay triangulations of a family of lattices that 
is obtained by distorting the integer grid in $\R^d$.
We show that the quality measures of this family are maximized for a certain 
distortion parameter, and that for this parameter, the lattice is isometric 
to the \emph{permutahedral lattice}, which is a well-known object in 
discrete geometry.
\end{abstract}

\section{Introduction}
\label{section:intro}

Simplicial meshes are now standard methods to approximate geometric
objects. These meshes are used in algorithms for several tasks,
including numerically solving partial differential equations, 
finite element approximation of functions and computational dynamical systems. 
The quality of
approximation of these algorithms depends on the {\em goodness} of the mesh. 
The notion of goodness of a triangulation is defined using some geometric 
properties of the simplices involved.
We discuss three measures to capture goodness: the first is the \emph{thickness}
of a simplex, which is the ratio of the minimal height to the maximal
edge length of the simplex.
The thickness of the triangulation is then the smallest thickness of any of
its simplices.
The second measure is the \emph{aspect ratio} of a simplex, which is the ratio
of the minimal height to the diameter of its circumsphere.
Again, the aspect ratio of a triangulation is the smallest aspect ratio of
its simplices.
For an introduction on the connections between simplex quality in a
mesh and finite element methods, we refer the reader to the survey of
Shewchuk \cite{Shewchuk-techreport2002}.

Delaunay triangulations are one of the most popular simplicial meshes. 
Delaunay triangulations are unique for a given (non-degenerate) point set and 
they satisfy many useful structural properties. 
Over the last few
decades fast algorithms have been developed to construct Delaunay triangulations
and to update them under insertions and deletions to the point set. 
They have been generalized to a larger class of 
{\em weighted Delaunay triangulations} \cite{Aurenhammer87}. 
Weighted point sets and point set perturbation techniques have been used to
get good Delaunay meshes in $\R^{2}$ and $\R^{3}$ both in theory and practice. 
For an introduction on these topics, we refer the reader to the recent book 
on Delaunay mesh generation by Cheng, Dey, and Shewchuk \cite{ChengDSbook13}.

In the context of Delaunay triangulations, recently Boissonnat, Dyer
and Ghosh \cite{bdg-delaunay} have introduced a new quality measure called
\emph{protection}: intuitively, this measures how far is a Delaunay 
triangulation from being degenerate. 
More specifically, consider the circumball of a $d$-simplex in the
Delaunay triangulation of a point set in general position in $\R^{d}$.
Since the point set is in general position, there are precisely $(d+1)$ points
incident to this ball, and none in the interior.
Protection is then defined as the maximal amount by which each circumball of a
$d$-simplex in the Delaunay triangulation can be inflated,
so that it does not contain any other point of the lattice in its interior.
In actual terms, protection is not a new measure but just a parameterization 
of the general position condition for point sets in Euclidean space. 
When it is clear from the context, we refer to the protection of a point set
as the protection of its Delaunay triangulation.
For any degenerate lattice, like the regular grid in $\R^d$, 
the protection value is $0$.
On the other hand, the \emph{permutahedral} lattice, which is one of the very
few lattices in general position, is known to have a high
value of protection~\cite{ckw-coxeter}. 
Boissonnat et al. \cite{bdg-delaunay} showed that good
protection implies stability of Delaunay triangulations with respect to 
metric distortions and perturbations of the point set. 
They also proved that a good value of protection guarantees good-quality
simplices~\cite{bdg-delaunay} in the Delaunay triangulation. 
This measure has also successfully been used to study Delaunay triangulations 
on manifolds, discrete Riemannian Voronoi diagrams, manifold
reconstructions from point sample and anisotropic meshing 
\cite{BoissonnatDG17,BoissonnatDGO17,BoissonnatSTY15,BoissonnatRW17,BoissonnatWY15}.

Algorithmic techniques (like weighted point sets, perturbation of
point sets, and refinement method) for getting good Delaunay triangulations work
well both in theory and practice in $\R^{2}$ and $\R^{3}$ 
\cite{ChengDEFTjacm00,EdelsbrunnerLMSTTUWstoc00,Litcs03a}. 
But these techniques do not scale well to higher dimensions. 
There are perturbation algorithms (see \cite{BoissonnatDGijcga14}) for points in $\R^d$,
that give quality measures such as thickness and aspect ratio of the order
$2^{-\Omega(d^{3})}$, which is exponentially small in $d$.
The same thing is true for protection in $\R^{d}$, see
\cite{BoissonnatDGijcga14,BoissonnatDGesa15}. 
This leads one to search for more structured points set in Euclidean space
whose Delaunay triangulations would have better quality guarantees. 
A natural class of candidates are lattices in $\R^{d}$. 

In this paper we concern ourselves with a family of lattices, which we obtain
by a distortion of the integer grid in $\R^d$ along the principal diagonal
direction $(1,\ldots,1)$. 
We call this family of lattices as the {\em diagonally distorted
lattices}. 
Essentially, stretching or squeezing the grid linearly along this direction gives 
this family of lattices.
This family was first studied in~\cite{ek-dual} by Edelsbrunner and Kerber,
with an ulterior motive to do topological analysis of high-dimensional
image data.
Later, this lattice was used to study covering and packing problems of
Euclidean balls in different contexts~\cite{ei-fcc, ek-diag,
  iku-sphere}.

All simplices in the Delaunay triangulation of a diagonally distorted lattice
are congruent.
Naturally, the thickness, aspect ratio and protection of each simplex is the same 
and defines the parameters for the lattice.
Recently, the qualities of a class of triangulations known as the 
\emph{Coxeter} triangulations was studied in~\cite{ckw-coxeter}.
This class of triangulations includes the Delaunay triangulation of the 
permutahedral lattice.

\subsection*{Outline of the paper}

In Section~\ref{section:backgnd} we discuss the preliminaries, including
protection, the permutahedral lattice and the family of diagonally distorted lattices,
and we summarize the main results of this paper.
Section~\ref{section:diagprotection} expands on the details of our results,
where we study the protection and other quality measures of the 
diagonally distorted lattices.

\section{Background and Contributions}
\label{section:backgnd}

We briefly mention a few geometric concepts needed for our results.
The interested reader may refer to
\cite{stanford-tech,bdg-delaunay,dutchbook,ckw-coxeter,csb-book,ek-dual}
for more details.

\subsection{General notations}

In this paper we work with the standard $\ell_{2}$-norm in
$\R^{d}$, and the distance between any two points $p,q$ in $\R^{d}$ will be 
denoted by $\distance{p}{q}$. 
For any point $p \in \R^{d}$  and any set $X \subset \R^{d}$, we denote the distance
between $p$ and $X$ as $\dist{p}{X}:= \inf_{x \in X} \distance{x}{p}$. 
Given any point $c\in \R^d$ and a radius $r \geq 0$, 
the ball $B(c, r) = \{ x\in \R^d : \dist{x}{c} < r \}$ is open,
and the ball $\overline{B}(c, r) = \{ x \in \R^d: \dist{x}{c} \leq r \}$ 
is closed.

For $X \subseteq \R^{d}$, the {\em convex hull} and {\em affine hull} of
$X$ will be denoted by $\mathrm{conv}(X)$ and $\mathrm{aff}(X)$, respectively.
For a set $X \subseteq \R^{d}$, we denote by $\dim(X)$ the affine dimension
of $\aff(X)$.

A simplex $\sigma= (p_{0}, \, \dots, \, p_{j})$
denotes the set of points $\{ p_{0}, \, \dots, \, p_{j}\} \subset \R^{d}$. 
The {\em combinatorial dimension}
of $\sigma$ is $j$, and {\em geometric dimension} of $\sigma$ is $\dim(\sigma)$.

A simplex $\tau$ is called a {\em sub-simplex} or a {\em face} (and
{\em proper face}) of a simplex $\sigma$ if $\tau \subseteq \sigma$
(if $\tau \subsetneq \sigma$). 
For any vertex $p_{i}$ in $\sigma = (p_{0}, \, \dots, \, p_{j})$,
$\sigma_{p_{i}}$
denotes the sub-simplex with vertex set $\{p_{0}, \, \dots, \, p_{j} \}
\setminus \{p_{i}\}$, and $D_{p_{i}}(\sigma)$ denotes the distance
$d \left(p_{i}, \aff\left(\sigma_{p_{i}}\right)\right)$.

We denote the circumradius and longest edge length
of $\sigma$ by $ R(\sigma)$ and $\eta(\sigma)$, respectively.
The quality measure {\em thickness} $\Theta(\sigma)$ of a simplex
$\sigma$ with combinatorial dimension $j$ is defined as 
\begin{equation}
\Theta(\sigma) =
  \begin{cases}
    1       & \quad \text{$j = 0$}\\
    \min_{p \in \sigma} \frac{D_{p}(\sigma)}{\eta(\sigma)}  & \quad \text{otherwise}
  \end{cases},
\end{equation}
and the {\em aspect ratio} $\Gamma(\sigma)$ is defined as 
\begin{equation}
\Gamma(\sigma) =
  \begin{cases}
    1       & \quad \text{$j = 0$}\\
    \min_{p \in \sigma} \frac{D_{p}(\sigma)}{2R(\sigma)}  & \quad \text{otherwise}
  \end{cases}.
\end{equation}

A lattice $\Lambda$ is a countable subset of $\R^{d}$ of the form 
$\Lambda := \left\{ \sum_{i=1}^{d} z_{i} v_{i} \,\mid\, \forall \, i,\,
z_{i} \in \mathbb{Z} \right\}$, 
where $\{v_{1}, \, \dots, \, v_{d}\}$ are linearly independent vectors in
$\R^{d}$. The vectors $\{v_{1}, \, \dots, \, v_{d}\}$ are called {\em representative vectors}
of $\Lambda$.  
We will interchangeably call points in $\Lambda$ as vectors to simplify notation.
Determinant $\det (\Lambda)$ of $\Lambda$, by abuse of notation, is the
absolute value of the determinant of the matrix whose columns are the vectors
$\{v_{1}, \, \dots, \, v_{d}\}$. Also, $s_{i}(\Lambda)$ will
denote the $i$-th smallest singular value of the matrix with columns
$\{v_{1}, \, \dots, \, v_{d}\}$.  Observe that $\det (\Lambda) = \prod_{i=1}^{d}s_{i}(\Lambda)$.

For a given lattice $\Lambda$, let $\lambda_{1}(\Lambda)$ denote the length of
smallest vector in $\Lambda$. 
The following result is a direct application of Minkowski's theorem \cite{csb-book}.

\begin{theorem}
\label{theorem:lambda-lattice-bounds}
For any lattice $\Lambda\subset \R^{d}$, we have 
$\lambda_{1}\left( \Lambda \right) \leq
\sqrt{d} \det(\Lambda)^{1/d}$. 
\end{theorem}

\subsection{Voronoi diagram, Delaunay complexes and protection}

Let $P$ be a subset of $\R^{d}$. 
For any point $p \in P$, {\em Voronoi cell} of $p$ is defined as the region
\[
\vor(p) : = \left\{ x \in \R^{d} \, \mid \, \forall q \in P, \, 
\distance{x}{p} \le \distance{x}{q} \right\},
\]
 and for a simplex $\sigma = (p_{0}, \, \dots, \, p_{k}) \subseteq P$, 
the {\em Voronoi cell}  of $\sigma$ is defined as $\vor(\sigma): =
\cap_{i=0}^{k} \vor(p_{i})$. The {\em Voronoi diagram} of $P$, denoted by
$\vor(P)$, is the decomposition of $\R^{d}$ into Voronoi cells of 
simplices with vertices from $P$. 
The {\em Delaunay complex} of $P$, $\del(P)$, is the {\em nerve} of
$\vor(P)$, that is, $\sigma \in \del(P)$ iff $\vor(\sigma) \neq
\emptyset$. 
For a point $p \in P$, the {\em star} of $P$, denoted by $\str(p, P)$, is the 
set of simplices $\sigma$ in $\del(P)$ such that $p \in \sigma$.  

Observe that for a lattice $\Lambda$, $\vor(\Lambda)$ (and $\del(\Lambda)$)
can be obtained by the periodic copies of 
the Voronoi cell (and star) of the origin $\vor(0)$ (and $\del(\Lambda)$). 
For the rest of the section, we denote by $\str(0)$ the star of the origin $\str(0,\Lambda)$,
when $\Lambda$ is clear from the context. 

First we formally state the notion of protection as defined in~\cite{bdg-delaunay}.
Consider a finite point set $P$ in $\R^d$.

\begin{defi}[$\alpha$-protection of a simplex]
 A simplex $\sigma\in \del(P)$
 is $\alpha$-protected if $\exists c\in \vor(\sigma)$ such that, for
 all $p \in \sigma$ and $q \in P\setminus \sigma$, we have
 $\distance{q}{c}\ge \distance{p}{c}+\alpha$.
\label{def:protection}
\end{defi}

\begin{defi}[$\alpha$-protection of a triangulation]
A triangulation $T$ of $P$ is said to be $\alpha$-protected for a non-negative
real $\alpha$, if
\[
\alpha=\mathrm{sup} \{\beta\ge 0 \mid \forall \sigma \in T, \sigma \text{ is $\beta$-protected}\}.
\]
\label{def:protection-triangulation}
\end{defi}

For a lattice $\Lambda$ in $\R^{d}$, observe that if all the
$d$-simplices in $\str(0)$ are $\alpha$-protected then 
$\del\left( \Lambda \right)$ is $\alpha$-protected. 

Adapting a result of Delaunay from \cite{Delaunay34} one can show that 
if all the
$d$-simplices in $\str(0)$ are $\alpha$-protected for some 
$\alpha > 0$, then $\del(\Lambda)$ is a triangulation of $\R^{d}$. 
A lattice $\Lambda$ is {\em degenerate} if $\del\left(\Lambda\right)$ is $0$-protected. 
Observe that if a lattice $\Lambda$ is degenerate then $\del(\Lambda)$ 
contains simplices with combinatorial dimension greater than $d$.   

Using \cite[Lemma 5.27]{boissonnat_chazal_yvinec_2018} and Theorem
\ref{theorem:lambda-lattice-bounds}, we get

\begin{theorem}
\label{theorem:lambda-protection}
Let $\Lambda$ be a lattice in $\R^{d}$ such that $\del(\Lambda)$ is
$\delta$-protected. Then
\[\delta \leq \lambda_{1}(\Lambda) \leq \sqrt{d} \det(\Lambda)^{1/d}.
\]
\end{theorem}

\subsection{Permutahedral Lattice}
\label{subsection:permuta}

Before we explore the Permutahedral lattice, we look at a closely related
lattice, the $A_d$ lattice: this is a $d$-dimensional lattice consisting
of the set of points $(x_1,\ldots,x_{d+1})\in \Z^{d+1}$ which satisfy 
$\sum_{i=1}^{d+1} x_i=0$. 
This lattice resides in the hyperplane $\sum_{i=1}^{d+1} y_i = 0$. 
Let us call this hyperplane $H$. 
One can observe that $A_d=H \cap \Z^{d+1}$.
For more details see~\cite{stanford-tech}. 

The $A_d^*$ lattice, also known as the Permutahedral lattice~\cite{csb-book} 
is the dual lattice to $A_d$. 
This means that it consists of points $\vec y = (y_1,\ldots,y_{d+1})\in H$ 
such that $\vec y\cdot\vec x \in \Z, \forall x\in A_d$.
Note that $A_d \subset A_{d}^*$, both lie in $H$, and contain the origin. 
The vertices of the Voronoi cell of the origin consist of all permutations of 
coordinates of the point 
\[s=\frac{1}{2(d+1)}(d,d-2,d-4,\ldots,-d+2,-d).\]
For this reason, this polytope is also called the \emph{permutahedron}
and lends the name \emph{permutahedral lattice} to $A_d^*$.
The permutahedron has precisely $(d+1)!$ vertices, each with the same norm. 
The representative vectors of $A_d^*$ lattice are of the form
\[
g_k=\frac{1}{d}( \underbrace{d+1-k,\ldots,d+1-k}_{k}, 
\underbrace{-k,\ldots,-k}_{d+1-k} ), 
\]
for $1\le k\le d$. 
This means that any point of $A_d^*$ can be expressed in the form 
$\sum m_k g_k$, where $ m_k\in\Z$ for all 
$k \in \{1, \, \dots, \, d\}$.

It was shown in~\cite{stanford-tech,ckr-polynomial} that $A_d^*$ is in general 
position, and that $Del(A_d^*)$ consists of congruent $d$-simplices. 
This means that it offers non-zero protection for its simplices.
Elementary calculations show that the Delaunay radius is
\[
R_{del}=\sqrt{\frac{d(d+2)}{12(d+1)}}.
\]
Let $R'$ denote the quantity 
\[
R':=\sqrt{\frac{d(d+2)}{12(d+1)}+\frac{2}{d+1}}
\]
Recently in~\cite{ckw-coxeter}, it was shown that

\begin{theorem}
 \label{theorem:astarpro}
 Protection for $Del(A_d^*)$ is 
 $Pro:=R'-R_{del}=
 \sqrt{\frac{d(d+2)}{12(d+1)}+\frac{2}{d+1}}-\sqrt{\frac{d(d+2)}{12(d+1)}}$.
\end{theorem}

\begin{corollary}
\label{cor:astarnormpro}
 The normalized protection of the $A^*_d$ lattice is the ratio of the
 protection to the Delaunay radius, that is,
 $Pro/R_{del}=\sqrt{\frac{d^2+2d+24}{d^2+2d}}-1=O\left(\frac{1}{d^2}\right)$.
\end{corollary}

\begin{remark}
\label{rem:powerpro}
 The power protection\footnote{For more details on power protection
 refer to~\cite{BoissonnatDGO17}.} 
 of $A^*_d$ is $R'^2-R_{del}^2=\frac{2}{d+1}$.
\end{remark}

\subsection{Diagonal distortion and Freudenthal triangulation}
\label{subsection:diagdist}

For a point set $P\in\R^d$, a diagonal distortion is a 
perturbation along the diagonal direction $(1,\ldots,1)$. 
Formally, the diagonal distortion of a vector 
$\vec x \in \R^d$ denoted by $T_{\delta}(x)$ is defined as:
\begin{equation}
\label{equation:dgds}
 T_{\delta}(\vec x) = \vec x - \left(\frac{1-\delta}{d}\right)\Delta(\vec x)\vec 1 ,
\end{equation}
where $\Delta(\vec x)=\sum_{i=1}^{d}x_i$, $\vec 1 = (1,\ldots,1)$ and 
$\delta\in \R$ is the distortion parameter.
This distortion was introduced by Edelsbrunner and Kerber in~\cite{ek-dual}, 
to build and study a family of lattices.

Here $\Delta^{-1}(0)$ denotes the hyperplane passing through the origin, 
which is normal to the vector $\vec 1$, that is, it is simply 
the hyperplane $H$.
$|\Delta(x)|$ is $\sqrt{d}$ times the height of $x$ from $H$.
Also for $\delta=1$, the linear transformation $T_{1}$ is the identity map, while for $\delta=0$,
it projects points on to $H$. For $0<\delta<1$, the transformation moves each point
closer to $H$, where the distance moved is proportional to the height of
the point from $H$, as evident from Equation~\eqref{equation:dgds}. 

In~\cite{ek-dual}, the authors built a family of lattices by setting $P:=\Z^d$.
Each $\delta\in\R$ gives a lattice, which we call a \emph{distorted grid}.
It is thus natural to talk about distorted cubes of the distorted grid,
which are images of a cube of $\Z^d$ under the transformation $T_{\delta}$.

\begin{figure}
\begin{center}
\includegraphics[width=0.350\textwidth]{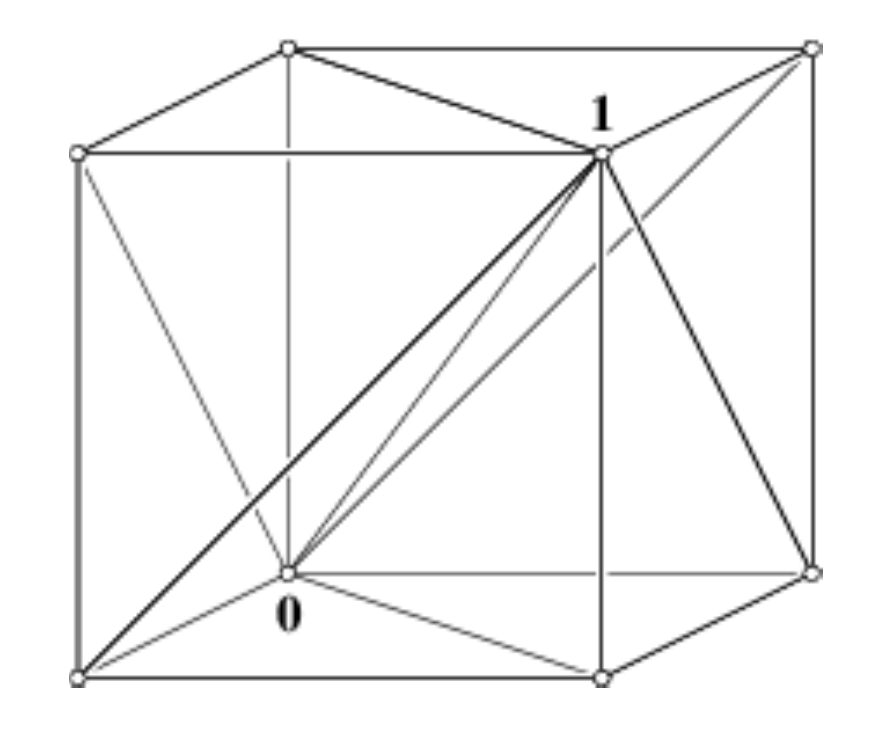}
\end{center}
\caption{Figure shows Freudenthal triangulation of the 3-cube,
figure from \cite{ek-diag}. In the figure $\mathbf{0} = \left( 0, \, 0,
  \, 0\right)$ and $\mathbf{1} = (1,\, 1, \, 1)$.}
\label{figure:freudenthal}
\end{figure}

Let $\square$ denote the $d$-cube $[0,1]^d$.
A \emph{monotone chain} on $\square$ is a sequence of a subset of 
its vertices such that
their coordinates are in strictly increasing lexicographic order.
More precisely, a sequence of vertices $(p_0,\ldots,p_k)\subset \{0,1\}^d$ is a 
monotone chain if for each pair $p_i=(v_1,\ldots,v_d),p_j=(w_1,\ldots,w_d)$ 
with $i<j$, it holds that $v_m\le w_m$ $\forall m\in[1,\ldots,d]$.
Each monotone chain can be interpreted as a simplex, which is the convex hull 
of its vertices.
It turns out that the collection of all simplices defined by monotone chains
triangulates the cube $\square$.
This triangulation is more commonly known as the 
\emph{Freudenthal triangulation}~\cite{freud-tgl} 
of the $d$-cube (also known as the Kuhn subdivision~\cite{kuhn-subdiv}).
This contains precisely $d!$ $d$-simplices~\cite{ek-dual}.
See Figure~\ref{figure:freudenthal} for a three-dimensional example.

In~\cite{ek-dual}, the authors show that for any $0<\delta\le 1$, 
the Delaunay triangulation of a distorted cube of $\Z^d$ is 
combinatorially equivalent to the Freudenthal triangulation 
of the unit $d$-cube. 
For $\delta=1$, the Delaunay triangulation of the distorted grid is degenerate, 
but for each $0<\delta<1$, it remains non-degenerate and combinatorially
stays the same~\cite{ek-dual}.

\subsection{Summary of Contributions}
\label{subsection:contributions}

\paragraph{Distorted grid and the permutahedral lattice}
Our first result is an interesting relation between the
distorted grids and the permutahedral lattices.
Since the distorted grid $T_0(\Z^d)$ resides in the hyperplane $H$, 
it is a $(d-1)$-dimensional point set. 
Also, we know that $A_{d-1}^*$ is a $(d-1)$-dimensional lattice, residing in $H$.
We show that

\begin{lemma}
\label{lemma:diagbasic}
 $T_0(\Z^d)$ is the $A_{d-1}^*$ lattice.
\end{lemma}

Lemma 2 in~\cite{ek-diag} shows that $T_0(\Z^d)$ is isometric
to $T_{\delta}(\Z^{d-1})$ for $\delta=1/\sqrt{d}$. 
Using Lemma~\ref{lemma:diagbasic} along
with this fact, we arrive at the conclusion that 
\begin{corollary}
\label{cor:diagiso}
$T_{\delta}(\Z^{d})$ is isometric to the $A^\ast_d$ lattice for $\delta=\frac{1}{\sqrt{d+1}}$.
\end{corollary}

In the light of the above result, we add a complementary observation:
\begin{lemma}
\label{lemma:otherlattice}
$T_{\delta}(\Z^{d})$ is isometric to the $A_d$ lattice 
for $\delta=\sqrt{d+1}$.
\end{lemma}

\paragraph{Protection for distorted grids}
We calculate the protection values for distorted grids, when the 
distortion parameter lies in the range $\delta\in (0,1]$.
Let $R_\delta$ denote the Delaunay radius for the parameter $\delta$. 
Then,

\begin{theorem}
\label{theorem:diagpro}
 The normalized protection values for the diagonally distorted lattice are
\begin{equation}
\Pro_\delta/R_\delta=
 \begin{cases}
 
 \sqrt{
 \frac{\delta^4(d^2-1)+\delta^2(d^2-22)+d^2+23}
 {\delta^4(d^2-1)+\delta^2(d^2+2)+d^2-1}
 } -1
  \approx \frac{24(1-\delta^2)}{d^2(\delta^4+\delta^2+1)}
 
 & \text{when $\frac{1}{\sqrt{d+1}} < \delta \le 1$}\\
 
 \sqrt{\frac{d^2+2d+24}{d^2+2d}}
 -1
 \approx \frac{24}{d^2}
 
 & \text{when $\delta=\frac{1}{\sqrt{d+1}}$}\\
 
 \sqrt{
 \frac{\delta^4(d^2-1)+\delta^2(d^2+24d+2)+d^2-1}
 {\delta^4(d^2-1)+\delta^2(d^2+2)+d^2-1}
 } -1
  \approx \frac{24\delta^2}{d(\delta^4+\delta^2+1)}
   
 & \text{when $0<\delta<\frac{1}{\sqrt{d+1}}$}\
 \end{cases}. 
\end{equation}
\end{theorem}

\paragraph{Thickness and Aspect ratio}
We further calculate the thickness and aspect ratio of the distorted
grid for $\delta\in (0,1]$.
\begin{theorem}
\label{theorem:thickaspect}
Let $\Theta_\delta$ denote the thickness, and $\Gamma_\delta$ denote
the aspect ratio of the distorted grid at parameter $\delta$.
Then,
\begin{equation}
\Theta_\delta=
 \begin{cases}
  \frac{1}{\delta\sqrt{2d}}
  & \text{for $1 \ge \delta\ge \frac{1}{\sqrt{2}}$}
  \\
  \frac{\sqrt{2-2\delta^2}}{\sqrt{d}} 
  & \text{for $\frac{1}{\sqrt{2}} \ge \delta \ge \frac{1}{\sqrt{d+1}}$}
  \\
  \frac{2\delta\sqrt{1-\delta^2}}{\sqrt{\delta^2d-\delta^2+1}} 
  & \text{for $\frac{1}{\sqrt{d+1}} \ge \delta > 0$}
 \end{cases}.
\end{equation}
and
\begin{equation}
 \Gamma_\delta=
 \begin{cases}
  \frac{\sqrt{3d}}{\sqrt{2}\sqrt{3\delta^2 d^2 + (1-\delta^2)^2(d^2-1)}}
  & \text{for $1 \ge \delta \ge \frac{1}{\sqrt{d+1}}$}
  \\
 \frac{\delta d\sqrt{3}}{\sqrt{\delta^2d-\delta^2+1}\sqrt{3\delta^2 d^2 + (1-\delta^2)^2(d^2-1)}} & 
 \text{for $\frac{1}{\sqrt{d+1}} \ge \delta > 0$} \\
 \end{cases}.
\end{equation}
\end{theorem}

One can see that the protection increases monotonically in the range 
$\delta\in \left(0,\frac{1}{\sqrt{d+1}}\right)$ and decreases monotonically to 0 in the range
$\delta\in \left(\frac{1}{\sqrt{d+1}},1\right]$.
The maximum protection value is attained at $\delta=\frac{1}{\sqrt{d+1}}$.
Similarly, the thickness and aspect ratio are also maximized for  
$\delta=\frac{1}{\sqrt{d+1}}$.

\begin{corollary}
From Corollary~\ref{cor:diagiso} we know that at parameter $\delta=\frac{1}{\sqrt{d+1}}$,
the distorted lattice is isometric to the $A^\ast$ lattice.
Thus, the values of protection, thickness and aspect ratio that are maximized
at this parameter agree with the results in~\cite{ckw-coxeter}, that was achieved
through an alternate analysis.	
\end{corollary}
In Figure~\ref{figure:plots} we plot the quality measures of the distorted grid
for a few dimensions.

\begin{figure}[H]
\centering
\includegraphics[width=0.425\textwidth]{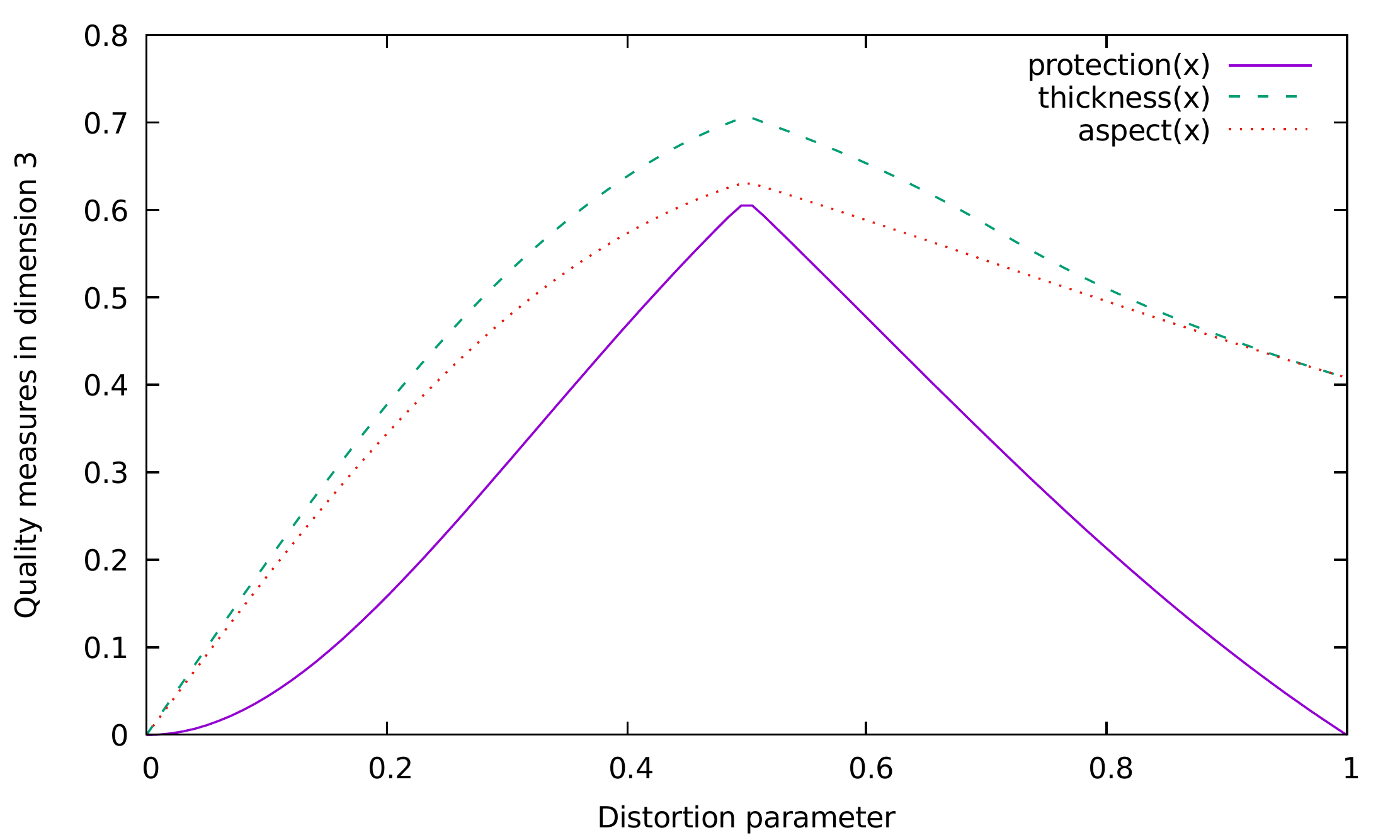}
\includegraphics[width=0.425\textwidth]{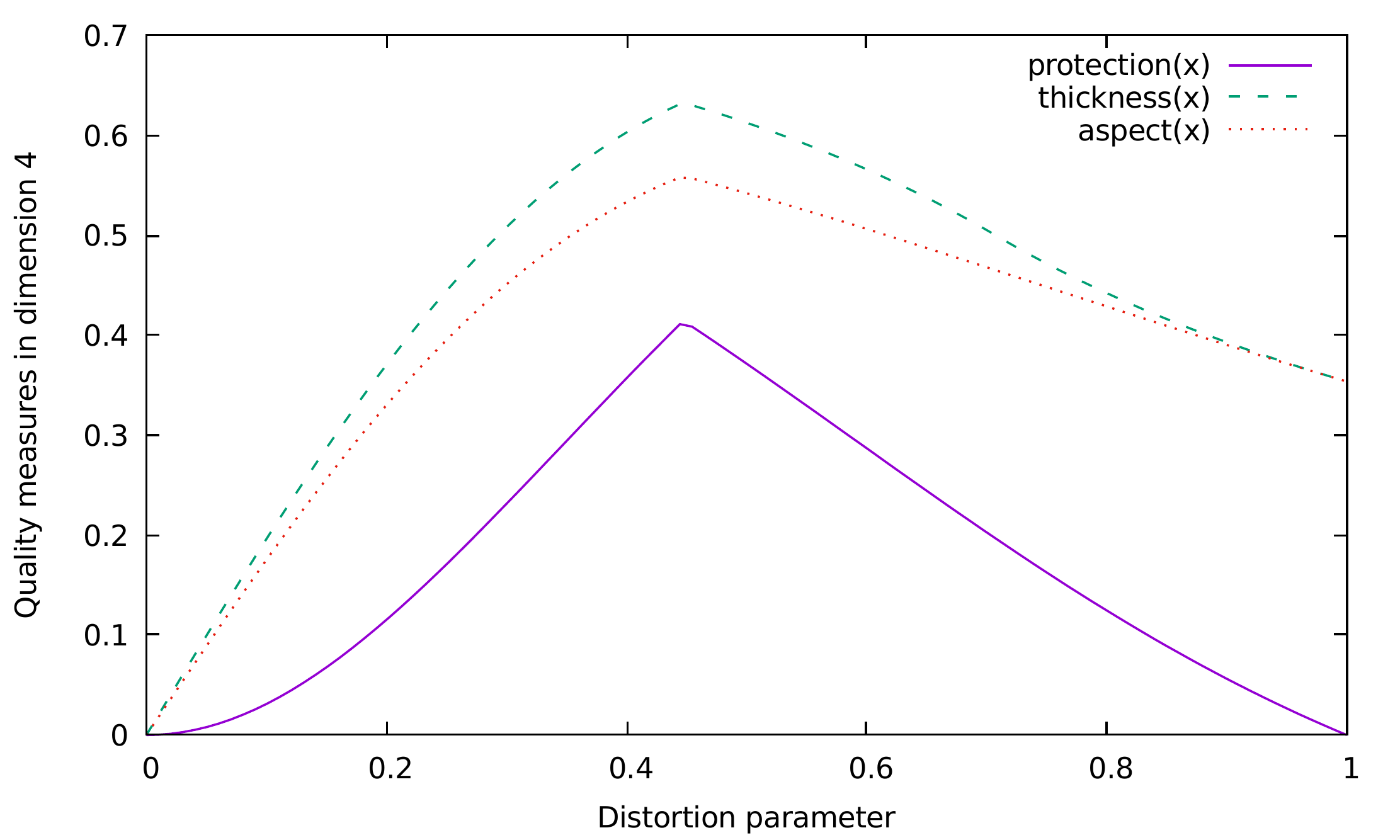}
\includegraphics[width=0.425\textwidth]{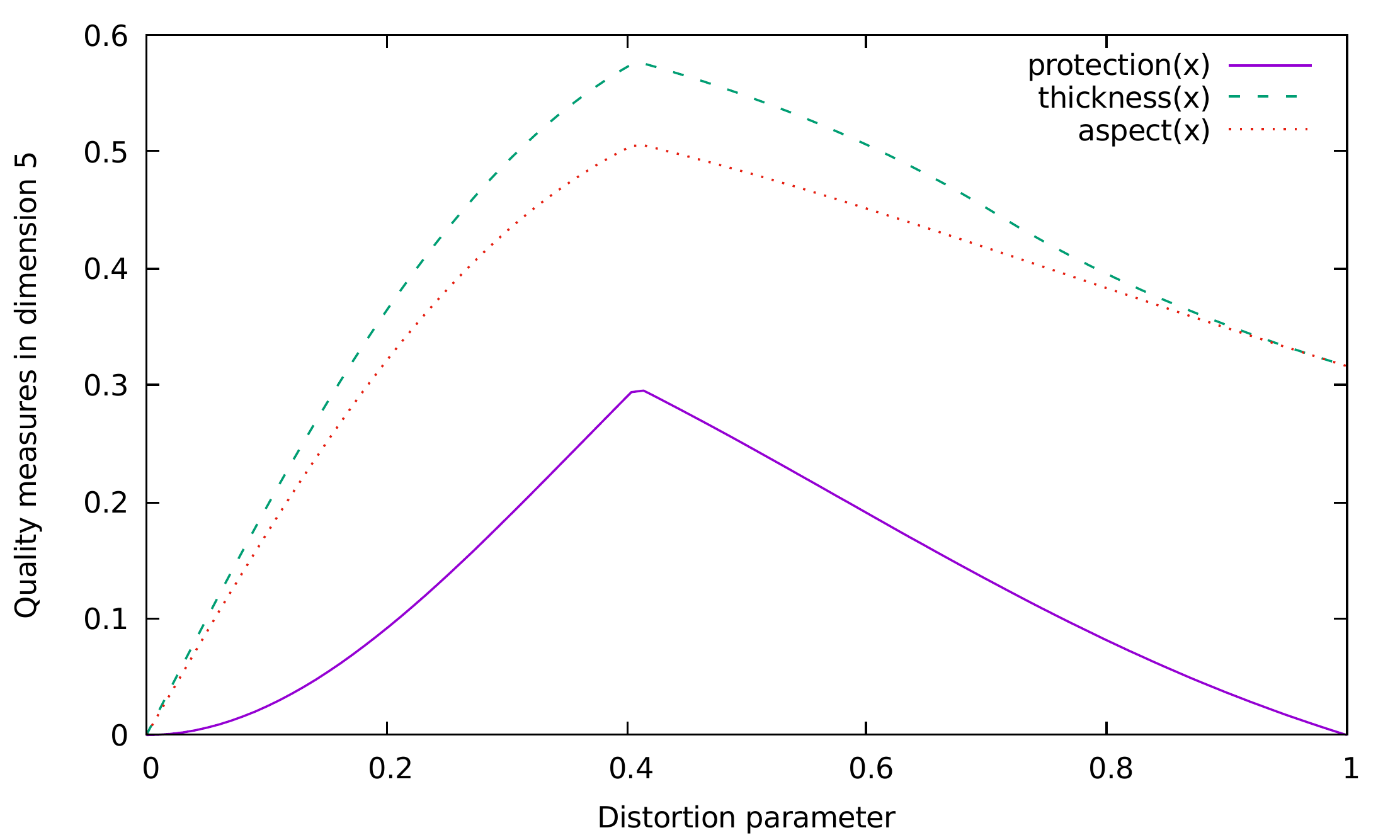}
\includegraphics[width=0.425\textwidth]{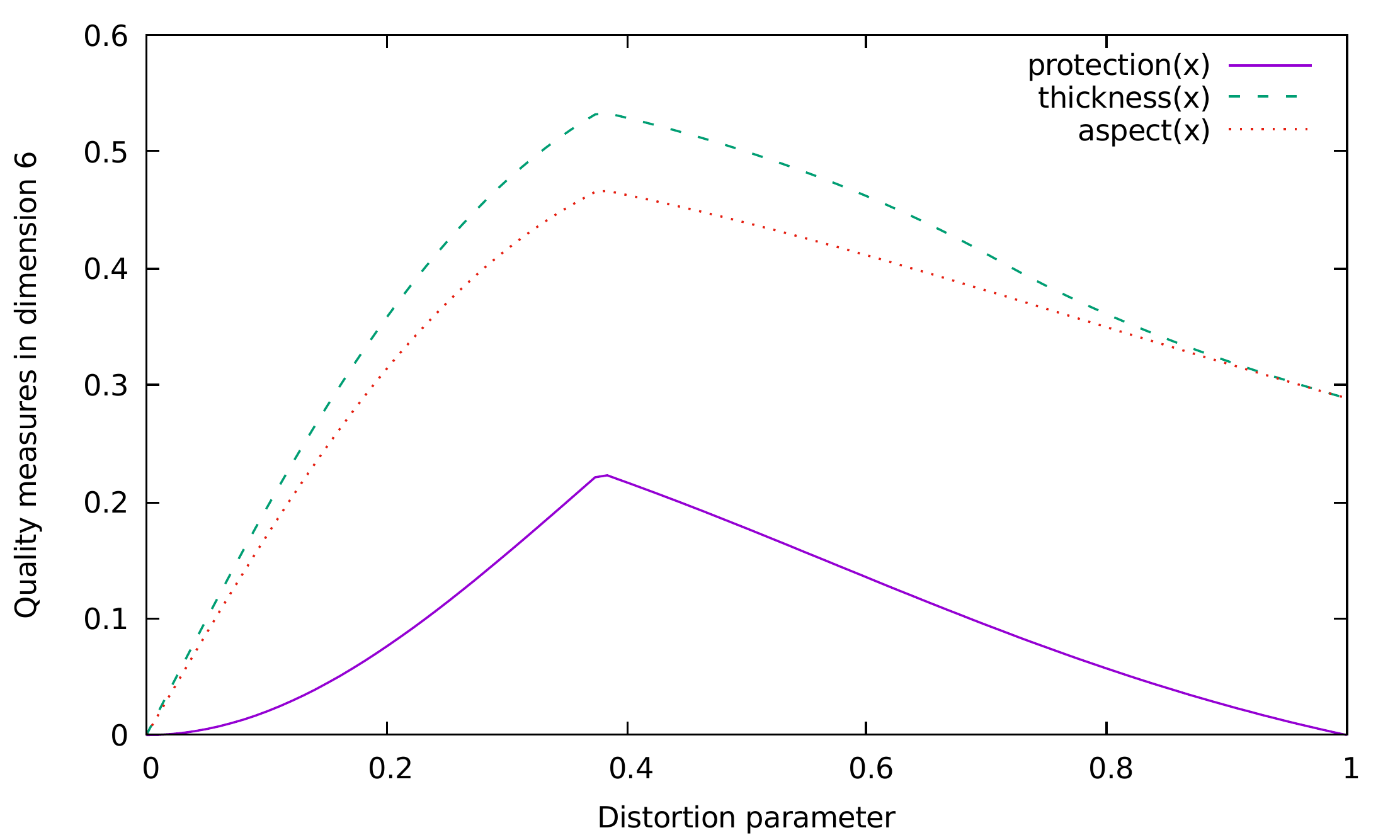}
\includegraphics[width=0.425\textwidth]{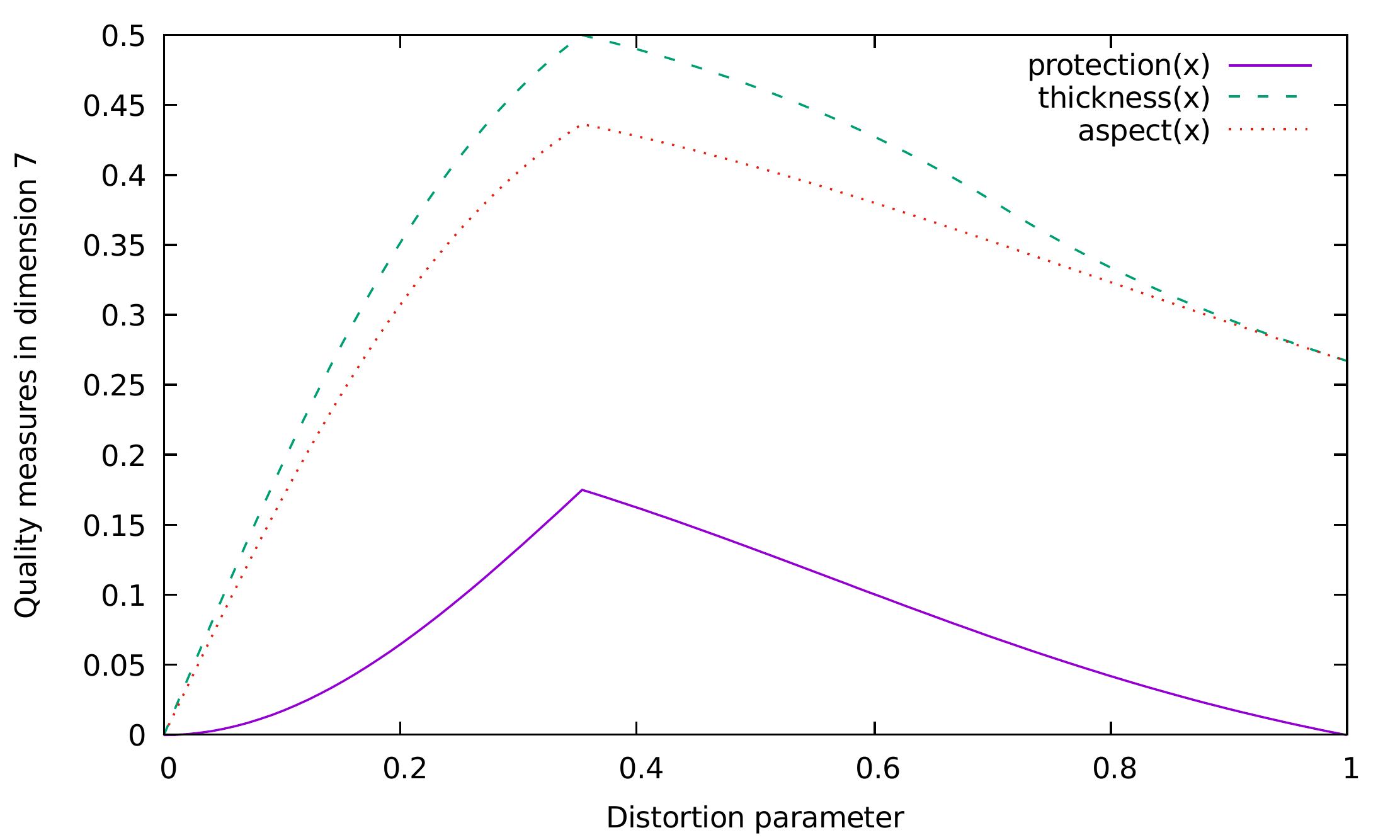}
\includegraphics[width=0.425\textwidth]{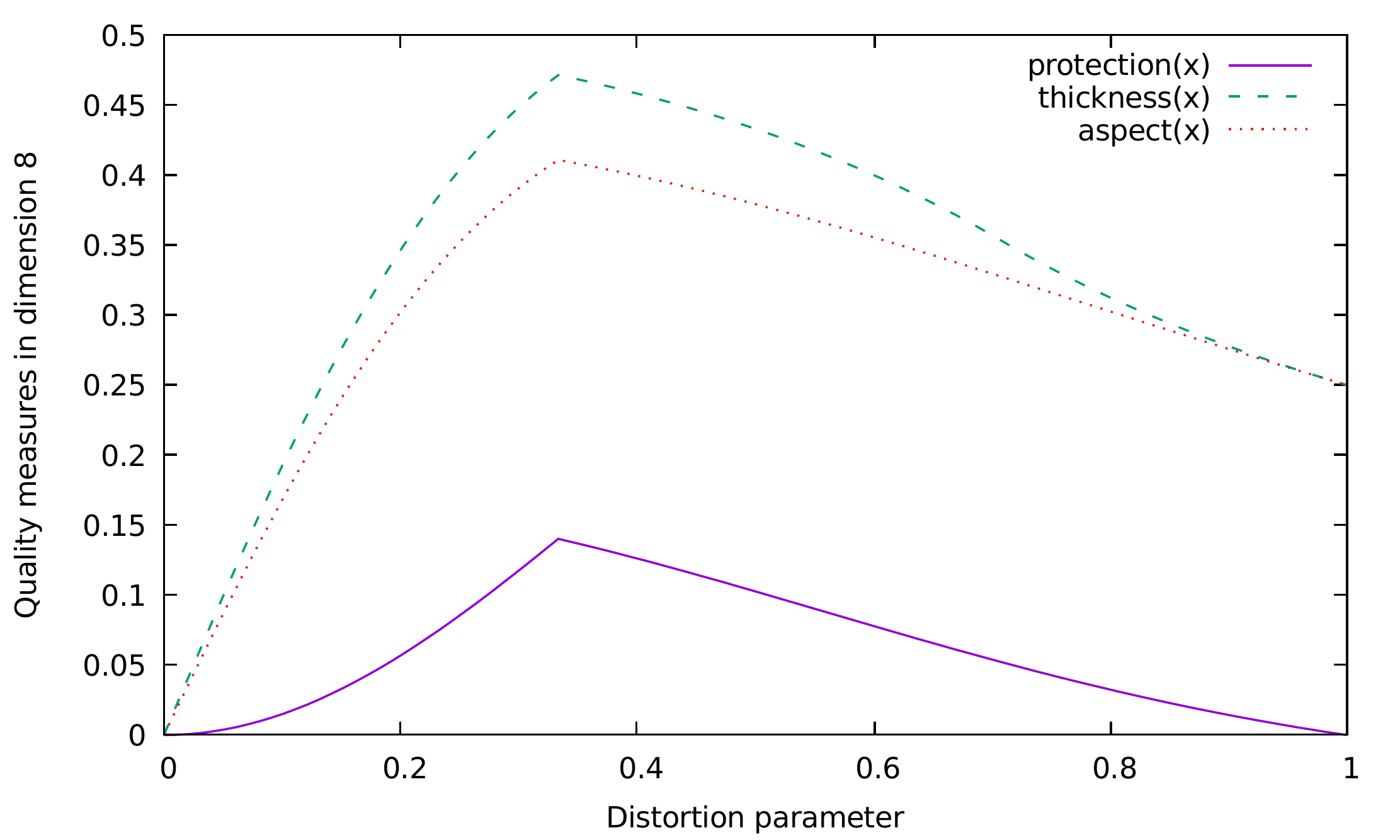}
\caption{Protection, thickness and aspect ratios of the distorted grid 
for a few dimensions.
Each parameter is maximized when the grid 
is distorted into the $A^\ast$ lattice.
}
\label{figure:plots}
\end{figure}
\section{Properties of Diagonally Distorted lattices}
\label{section:diagprotection}

\subsection{Proof of Lemma~\ref{lemma:diagbasic}}

\begin{proof}
We prove the claim in two steps, first by showing that 
$T_0(\Z^d) \subseteq A_{d-1}^*$, and then showing that 
$A_{d-1}^*  \subseteq T_0(\Z^d)$, which implies the result.

$T_0(\Z^d) \subseteq A_{d-1}^*$: consider any point $x=(x_1,\ldots,x_d)\in\Z^d$. 
Let $\vec x'= T_0(x)=\vec x - \frac{1}{d}\Delta(x)\vec 1$, 
which means that $x'_i = x_i-\frac{\sum_{j=1}^{d} x_j}{d}$, $\forall \, 1\le i\le d$.
Recall the definition of $A_{d-1}^*$ lattice: it consists of all points 
$y\in \R^d$ such that $\vec y\cdot \vec z$ is an 
integer for each $\vec z\in A_{d-1}$. 
Also, $\Delta(\vec z)=0$ by definition.
Now consider the dot product $\vec x'\cdot \vec z$,
\begin{eqnarray*}
\vec x'\cdot \vec z 
  &=& \sum_{i=1}^{d} \left(x_i - \frac{\sum_{j=1}^{d} x_j}{d}\right)\cdot z_i \\
  &=& \sum_{i=1}^{d} \left(x_i\cdot z_i -\frac{\sum_{j=1}^{d} x_j}{d}\cdot z_i \right)\\
  &=& \sum_{i=1}^{d} x_i\cdot z_i -\left(\frac{\sum_{j=1}^{d} x_j}{d}\right) \sum_{i=1}^{d} z_i \\
  &=& \sum_{i=1}^{d} x_i\cdot z_i,
\end{eqnarray*}
which is an integer since both $\vec{x},\vec{z}\in\Z^d$. 
This holds for all points $\vec z\in A_{d-1}$,
so $x'$ is a point of $ A_{d-1}^*$.
Since each point $x'\in T_0(\Z^d)$ satisfies the membership criteria 
for $A_{d-1}^*$, it holds that  $T_0(\Z^d) \subseteq A_{d-1}^*$.

$A_{d-1}^*  \subseteq T_0(\Z^d)$: for any point $x\in A_{d-1}^*$ to be 
a member of $T_0(\Z^d) $, it must have a special form: 
there must exist a point $X=(X_1,\ldots,X_d)\in\Z^d$ such that $x$ is the
projection of $X$ onto $H$.
Specifically, $x_i=X_i - \frac{\sum_{j=1}^{d} X_j}{d}$, 
must hold for each $i\in\{1,\ldots,d\}$.
We show that $x$ has this special form.

The representative vectors of $A_{d-1}^*$ lattice~\cite[Chap4.]{csb-book} 
are of the form
\[
g_i=\frac{1}{d}( \underbrace{d-i,\ldots,d-i}_{i}, 
\underbrace{-i,\ldots,-i}_{d-i} ), 
\]
for $1\le i \le d-1$.
So $x$ can be uniquely written as the linear combination 
$x=\sum_{i=1}^{d-1}k_i g_i$, where each $k_i$ is an integer. 
Consider the point of $A_d$, $X=(X_1,\ldots,X_d)$ such that 
$X_j=\sum_{i=j}^{d-1}k_i $ for $j\in \{1,\ldots,d\}$ 
(note that $X_d=0$).
We show that $x=T_0(X)$, which proves the claim.

Expanding $x$, we get
\begin{eqnarray*}
x = && \frac{k_1}{d} (d-1,-1,-1,\ldots,-1)+ \\
    && \frac{k_2}{d} (d-2,d-2,-2,\ldots,-2)+ \\
    && \ldots \\
    &&\frac{k_{d-2}}{d} (2,\ldots,2,-2+d,-2+d)+ \\
    &&\frac{k_{d-1}}{d} (1,\ldots,1,-1+d) .
\end{eqnarray*}
Simplifying, we see that 
$x_1=\frac{1}{d}\left(d k_1 - k_1 + d k_2 - 2 k_2 + \ldots + d k_{d-1} - 
(d-1) k_{d-1}\right) $ 
which simplifies to
\[
\frac{1}{d}\{d(k_1 + k_2 +\ldots +k_{d-1})-
\{k_1 + 2 k_2 + \ldots + (d-1) k_{d-1}\}\} = 
\sum_{i=1}^{d-1}k_i-\frac{\sum_{q=1}^{d-1}q k_{q}}{d}. 
\]
Similarly, it is easy to see that 
$x_j=\sum_{i=j}^{d-1}k_i-\frac{\sum_{q=1}^{d-1} q k_{q}}{d}$. 
Hence, $x$ is of the form 
\[
\left(X_1 - \frac{\sum_{i=1}^{d} X_i}{d},X_2 - \frac{\sum_{i=1}^{d} X_i}{d},\ldots,X_{d-1} - 
\frac{\sum_{i=1}^{d} X_i}{d}, X_d-\frac{\sum_{i=1}^{d} X_i}{d}\right)
=T_0(X).
\]
\end{proof}

\paragraph{Proof of Lemma~\ref{lemma:otherlattice}}
The proof idea is to find a bijection between the basis vectors of the
two lattices $T_\gamma(\Z^d)$ and $A_d$ for $\gamma=\sqrt{d+1}$,
such that the bijection preserves norms of the vectors
and the dot products between them.

Specifically, we choose a basis for $T_\gamma(\Z^d)$ as 
$\{T_\gamma(e_1),\ldots,T_\gamma(e_d)\}$, where $\{e_i\}_{i=1\ldots d}$
is the standard basis for $\R^d$.
For $A_d$, the standard basis in $\R^{d+1}$ is 
\[
\{u_1,\ldots,u_d\}:=\{(1,-1,0,\ldots,0),(0,1,-1,0,\ldots,0),\ldots,
(0,\ldots,0,1,-1)\}.
\]
The bijection takes $T_\gamma(e_i)$ to $u_i$ for each $i$.
It is easy to calculate that the norm of each vector is $\sqrt{2}$.
Moreover, it can be verified that 
$T_\gamma(e_i)\cdot T_\gamma(e_j) = u_i\cdot u_j$
for all $1\le i,j\le d$.
As a result, there is an bijection between the two lattices, 
that takes the point $p=\sum_{i=1}^{d}m_i T_\gamma(e_i)$ of $T_\gamma(\Z^d)$ 
where $m_i\in \Z$, to the point $q=\sum_{i=1}^{d}m_i u_i$ of $A_d$.
This bijection is an isometry because of the above conditions.

\subsection{Protection of distorted grids}
\label{subsection:dgpro}

In this sub-section we prove Theorem~\ref{theorem:diagpro} by
studying the protection properties of the family of lattices 
$\{ T_{\delta}(\Z^d) \}_{0< \delta \le 1}$.
Throughout this sub-section, we will assume that $\delta$ is a value in 
this range, if it is not stated explicitly.

Let $\Pro_\delta$ denote the protection of $T_\delta(\Z^d)$. 
For $\delta=1$, the Delaunay triangulation of $T_\delta(\Z^d)$ 
is degenerate and hence offers 0 protection, so $\Pro_1=0$.
We calculate $\Pro_\delta$ for $\delta \in [0,1]$ by first calculating the
protection of a specific simplex in the Delaunay triangulation 
and then showing that all simplices have the same protection.

Let $\{ e_1,\ldots, e_d \}$ denote the standard basis of $\R^d$.
Let $\sigma$ denote the simplex $\sigma=(v^0,\ldots,v^d)$ 
where $v^i$ is the vector sum $v^i=\sum_{j=d+1-i}^{d}e_j$.
That means, $v^0=\underbrace{(0,\ldots,0)}_{d}$, 
$v^d=\underbrace{(1,\ldots,1)}_{d}$ and more generally 
$v^i=(\underbrace{0,\ldots,0}_{d-i},\underbrace{1,\ldots,1}_{i})$.
It is straightforward to see that the simplex $\sigma$ is a part of the 
Freudenthal triangulation of the $d$-cube $[0,1]^d$.
We inspect the distorted transformation of $\sigma$, which we denote by 
$\sigma_{\delta}=T_{\delta}(\sigma)$; this is a Delaunay simplex of
the distorted grid $T_{\delta}(\Z^d)$. 
We calculate the protection for $\sigma_{\delta}$ to determine the value
for $T_{\delta}(\Z^d)$.
To do so, we first find the circumcentre and circumradius of the simplex.

The vertices of the simplex can be written as 
$\sigma_\delta=(v^0_\delta,\ldots,v^d_\delta)$,
where $v^i_\delta=T_{\delta}(v^i)$, $1\le i\le d$. 
In particular, 
\[
v^i_\delta=\frac{1}{d} (\underbrace{-i+i\delta,\ldots,-i+i\delta}_{d-i},
\underbrace{d-i+i\delta,\ldots,d-i+i\delta}_{i}).
\]
We denote the circumcentre of $\sigma_{\delta}$ as $C_{\delta}$ 
and the radius of the circumsphere by $R_{\delta}$.
Setting $\delta=0$, we see that 
$v^i_\delta=\frac{1}{d}(\underbrace{-i,\ldots,-i}_{d-i},
\underbrace{d-i,\ldots,d-i}_{i})$.
It follows that $T_{0}(\sigma)$ is a Delaunay simplex of $A^{\ast}_{d-1}$.
The circumcentre of $\sigma_0$ is 
$C_0:=\frac{1}{2d}(-d+1,-d+3,\ldots,d-3,d-1)$, which is a Voronoi vertex 
of $A^{\ast}_{d-1}$.
Then, the radius of the circumsphere is 
$R_0=\sqrt{\frac{d^2-1}{12d}}$~\cite{ek-diag}.
On the other hand, for $\delta=1$, the simplex is $\sigma_1=\sigma$.
Then the circumcentre is $C_1=(1/2,\ldots,1/2)$ and $R_1=\sqrt{d}/2$.
For intermediate values of $\delta$, from~\cite{ek-diag}
we have the relation that 
\[
R_\delta=\sqrt{\delta^2 R_1^2 + (1-\delta^2)^2 R_0^2},
\]
which can be simplified as 
$R_\delta^2=\frac{\delta^2 d}{4}+\frac{(1-\delta^2)^2(d^2-1)}{12d}.$

We now calculate $C_\delta$.
Note that $C_\delta$ is equidistant from each $v^i_\delta$.
Since $v^0_\delta=(0,\ldots,0)$, we have that for all $1\le i\le d$, 
$(C_\delta-v^0_\delta)^2=C_\delta^2=(C_\delta-v^i_\delta)^2$.
This simplifies to $\dotp{C_\delta}{v^i_\delta}=\frac{|v^i_\delta|^2}{2}$.
So we have a set of $d$ equations,
\begin{eqnarray}
 \left[
 \begin{array}{c}
  \vec{v^1_\delta}\\
  \ldots \\
  \vec{v^d_\delta}\\
 \end{array}
 \right]C_\delta=
 \frac{1}{2}\left[
 \begin{array}{c}
  |\vec{v^1_\delta}|^2\\
  \ldots \\
  |\vec{v^d_\delta}|^2\\
 \end{array}
 \right],
\end{eqnarray}
where $C_\delta$ is a column vector and the rows of the matrix on the left hand
side contain the vertices of the simplex.
Solving the system of equations, it follows that 
\[
C_\delta=\delta C_1 + (1-\delta^2)C_0.
\]
This can be written as
\[
C_\delta=\left[\frac{\delta}{2}+\frac{(1-\delta^2)(-d+1)}{2d},
\frac{\delta}{2}+\frac{(1-\delta^2)(-d+3)}{2d},\ldots,
\frac{\delta}{2}+\frac{(1-\delta^2)(d-1)}{2d}\right].
\]

\begin{remark}
\label{remark:barycentric}
$C_\delta$ can be explicitly written as the barycentric coordinates of 
$\sigma_\delta$ as $C_\delta=\sum_{i=0}^{d}\mu_i v^i_\delta$, where 
$\mu_0,\mu_d=\frac{1+(d-1)\delta^2}{2d}$ and  
$\mu_1,\ldots,\mu_{d-1}=\frac{1-\delta^2}{d}$.
Note that $\forall i, \mu_i > 0$. So the simplex is well-centered, that is,
the circumcenter lies in the interior of the simplex.
\end{remark}

\paragraph{Candidates for protection}
To calculate the protection for $\sigma_\delta$, we find the vertices of $\Z^d$  
which after distortion realize the protection value for $\sigma_\delta$.
Since $\sigma$ is a $d$-simplex, it has $(d+1)$ facets.
Consider such a facet $f$. 
In the Freudenthal triangulation, the facet $f$ has two $d$-simplices 
as co-faces, one being $\sigma$.
Both these simplices are formed by adding a vertex to $f$.
Let $v$ be the vertex of $\sigma$ that when added to $f$ forms $\sigma$,
and let $p$ be the vertex of $\Z^d$ that forms 
the other $d$-simplex with $f$.
We say that $p$ is opposite to $v$.

Each vertex of $\sigma$ has an opposite vertex in $\Z^d$.
We call the set of such opposite vertices as $n(\sigma)$.
Since there are $(d+1)$ facets, so $|n(\sigma)|=d+1$.
Note that since the combinatorial structure of the triangulation does not change 
with a change in $\delta$, $n(\sigma_\delta)=T_\delta(n(\sigma))$.
First we calculate the protection offered by the distortions of the 
points of $n(\sigma)$, and then show that these are precisely the points 
of the distorted grid which define the protection for $\sigma_\delta$.

Given a vertex $v^i$ of $\sigma$, we denote by $p^i$ 
the opposite vertex in $n(\sigma)$.
Using the monotone chain property of Freudenthal triangulation,
it follows that
\begin{equation}
 p^i=
 \begin{cases}
  (2,1,\ldots,1) & i=0 \\
  (\underbrace{0,\ldots,0}_{d-i-1},1,0,\underbrace{1,\ldots,1}_{i-1}) 
	& 1\le i \le d-1 \\
  (0,\ldots,0,-1) & i=d \\
 \end{cases}.
\end{equation}
Note that each $p_i$ is distinct.
We see that $\Delta(p^0)=d+1$, $\Delta(p^i)=i$ for $1\le i \le d-1$, 
and $\Delta(p^d)=-1$, so that
\begin{equation}
 T_\delta(p^i)=
 \begin{cases}
  p^0-\frac{(1-\delta)}{d}(d+1)\vec{1} & i=0 \\
  p^i-\frac{(1-\delta)}{d}(i)\vec{1} & 1\le i \le d-1 \\
  p^d+\frac{(1-\delta)}{d}\vec{1} & i=d \\
 \end{cases}.
\end{equation}
We next define $D_i=|T_\delta(p_i)-C_\delta|-R_\delta$ as the protection 
offered by the point $T_\delta(p^i)$. 
After calculations, we see that $D_0=D_d$ and $D_1=\ldots=D_{d-1}$ for all $\delta$.
Also,
\begin{equation}
 (D_0=D_d)
 \begin{cases}
  > (D_1=\ldots=D_{d-1}) & \text{when $\frac{1}{\sqrt{d+1}} < \delta \le 1$}\\
  = (D_1=\ldots=D_{d-1}) & \text{when $\delta=\frac{1}{\sqrt{d+1}}$}\\
  < (D_1=\ldots=D_{d-1}) & \text{when $0< \delta <\frac{1}{\sqrt{d+1}} $}\\
 \end{cases}.
\end{equation}
This agrees with the observation that for $\delta=\frac{1}{\sqrt{d+1}}$, 
the distorted lattice is isometric with the 
$A^\ast_d$ lattice (Corollary ~\ref{cor:diagiso}) and hence has $(d+1)$ 
points defining the protection (see~\cite{ckw-coxeter}).
So, we can define the offered protection as
\begin{equation}
 \Pro_\delta =
 \begin{cases}
  D_1=\ldots=D_{d-1} & \text{when $\frac{1}{\sqrt{d+1}} < \delta \le 1$}\\
  D_0=D_1=\ldots=D_{d-1}=D_d & \text{when $\delta=\frac{1}{\sqrt{d+1}}$}\\
  D_0=D_d & \text{when $0<\delta<\frac{1}{\sqrt{d+1}}$}\
 \end{cases}.
\end{equation}
The explicit values are
\begin{equation}
\label{equation:rawprotection}
\Pro_\delta=
 \begin{cases}
 
 \sqrt{\frac{\delta^4(d^2-1)+\delta^2(d^2-22)+d^2+23}{12d}}
 -\sqrt{\frac{\delta^4(d^2-1)+\delta^2(d^2+2)+d^2-1}{12d}}
 
 & \text{when $\frac{1}{\sqrt{d+1}} < \delta \le 1$}\\
 
 \sqrt{\frac{d(d+2)}{12(d+1)}+\frac{2}{d+1}}
 -\sqrt{\frac{d(d+2)}{12(d+1)}}
 
 & \text{when $\delta=\frac{1}{\sqrt{d+1}}$}\\
 
 \sqrt{\frac{\delta^4(d^2-1)+\delta^2(d^2+24d+2)+d^2-1}{12d}}
 -\sqrt{\frac{\delta^4(d^2-1)+\delta^2(d^2+2)+d^2-1}{12d}}
 
 & \text{when $0<\delta<\frac{1}{\sqrt{d+1}}$}\
 \end{cases}.  
\end{equation}

Similarly, we define the power protection offered by $T_\delta(p^i)$ as 
$E_i=|T_\delta(p_i)-C_\delta|^2-R_\delta^2$.
Substituting these values in $E_i$, we see that
\begin{remark}
The power protection values are
\begin{equation}
 E_i=
 \begin{cases}
  2\delta^2 & i=0 \\
  \frac{2}{d}(1-\delta^2) & 1\le i \le d-1 \\
  2\delta^2 & i=d \\
 \end{cases}.
\end{equation}
\end{remark}

Next, we show that
\begin{lemma}
 \label{lemma:propoint}
 For any $0<\delta<1$, vertices of $n(\sigma_\delta)$ determine 
 the protection for $\sigma_\delta$.
\end{lemma}
\begin{proof}
We prove the claim by contradiction. 
Assume that some point $p\not\in n(\sigma_\delta)$ of the lattice determines
the protection for $\sigma_\delta$, that is, it is the closest lattice point 
to $C_\delta$, not counting $\sigma_\delta$'s vertices.
Since all Delaunay $d$-simplices are congruent and have congruent 
neighborhoods, it follows that $p$ and $\sigma_\delta$ is the 
minimal configuration defining the protection.

Let $B(\sigma_\delta)$ be $\sigma_\delta$'s circumball. 
Let $q$ be the projection of $p$ on $B(\sigma_\delta)$.
Next, consider the circumballs of simplices formed using a facet of 
$\sigma_\delta$ and the corresponding reflected opposite vertex 
in $n(\sigma_\delta)$.
These circumballs together form a shell around $B(\sigma_\delta)$, since they 
each cover all $(d-1)$-facets of $\sigma_\delta$ and hence cover the 
boundary of $B(\sigma_\delta)$ completely. 
Since we have Delaunay simplices, $p$ lies outside this collection 
of circumballs.
We show that $q$ must be a vertex of $\sigma_\delta$. 
Suppose it is not the case. 
Then, $q$ lies in the interior of at least one such Delaunay ball $B'$.
In that case, $p$ is closer to the boundary of $B'$ than $B(\sigma_\delta)$.
That means the protection of $p$ with $B'$ is lower than that 
with $\sigma_\delta$, which violates our assumption that $p$ and 
$B(\sigma_\delta)$ is the minimal configuration defining the protection. 
So, $q$ does not lie in the interior of any Delaunay circumball.
This is only possible if $q$ is a vertex of $\sigma_\delta$. 
An example is shown in Figure~\ref{fig:pro2d}.
\begin{figure}
\centering
\includegraphics[scale=0.3]{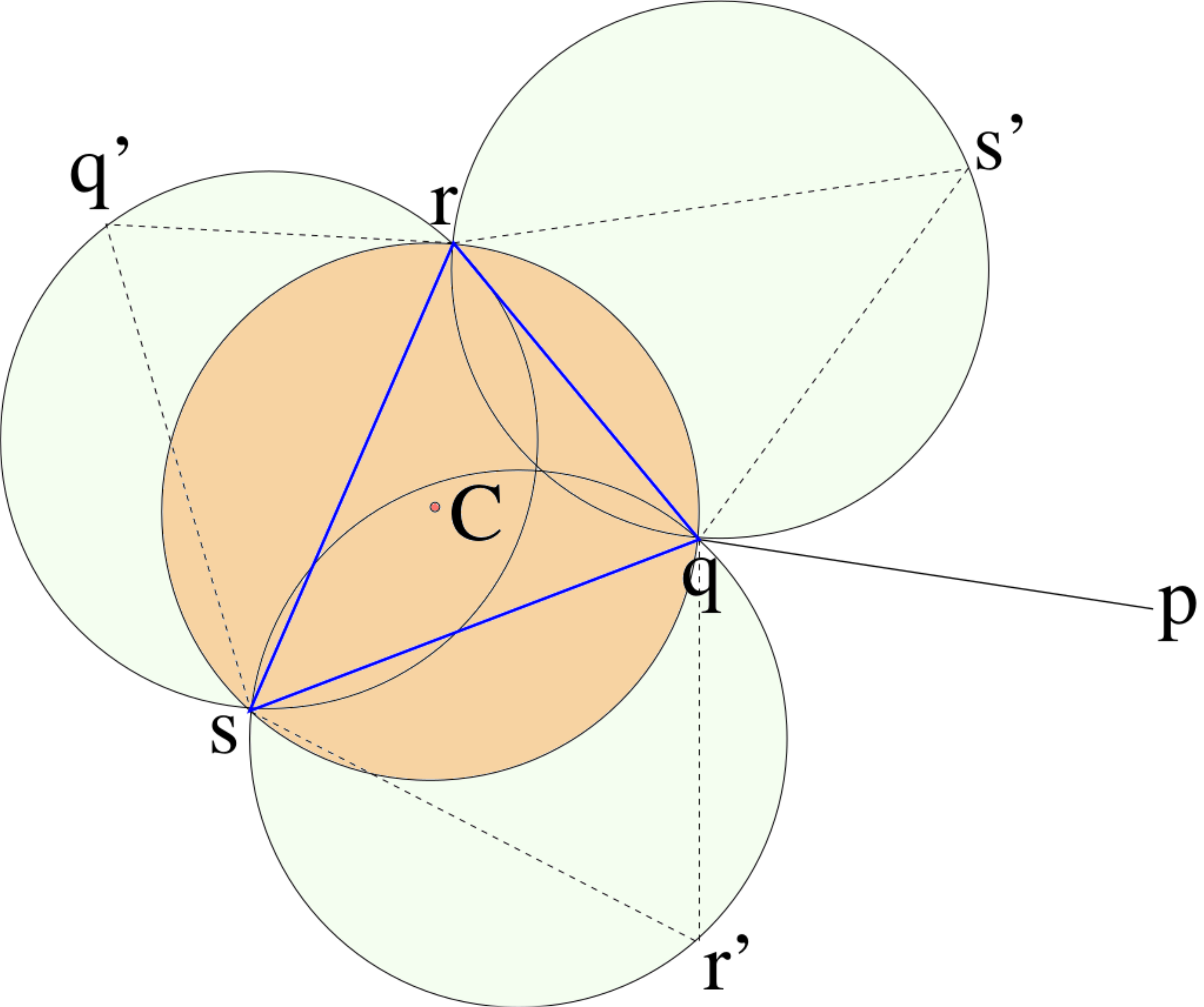}
\caption{A Delaunay simplex $\sigma_\delta=(q,r,s)$ in the plane 
with  circumcenter $C$. 
The vertex $r'$ is the reflection of $r$ along the edge $qs$. 
Reflections along edges $rs$ and $qr$ gives the vertices
$q'$ and $s'$, respectively. 
The circumballs of $(r',q,s)$, $(s',r,q)$ and $(q',s,r)$ completely
cover $\sigma_\delta$. 
The closest point to $p$ on $\sigma_\delta$'s circumball is a vertex 
$q$ of $\sigma_\delta$.
}
\label{fig:pro2d}
\end{figure}

Since $p$ and $q$ are both lattice points, we can write 
$p-q=\sum_{j=1}^{j=d}m_j v^{j}_{\delta}$, where $m_j$s are integers. 
Since $(p,q)$ is orthogonal to $B(\sigma_\delta)$, we have that 
$C_\delta$, $q$ and $p$ are collinear, where $C_\delta$ is the centre 
of $B_\delta$.
Without loss of generality, assume that $q=v_\delta^i$ for $i\in \{0,\ldots,d\}$.
Because of collinearity, we have that
$C_\delta-q= \lambda(p-q)= 
\lambda(\sum_{j=1}^{j=d} m_j v_\delta^j-v_\delta^i)$, or 
$C_\delta=\lambda \sum_{j=1}^{j=d} m_j v_\delta^j + (1-\lambda) v_\delta^i$
for a real value $\lambda$. 
We consider two cases:
\begin{itemize}
\item $i=\{0,d\}$: 
we look at the case when $i=0$. 
The argument for $i=d$ is very similar.
For $i=0$ we have $C_\delta=\lambda \sum_{j=1}^{j=d} m_j v_\delta^j =\lambda p$.
We recall from Remark~\ref{remark:barycentric} that 
$C_\delta=\sum_{k=0}^{d}\mu_k v^k_\delta$, 
where $\mu_0,\mu_d=\frac{1+(d-1)\delta^2}{2d}$
and $\mu_1,\ldots,\mu_{d-1}=\frac{1-\delta^2}{d}$.
This immediately gives us that 
$\lambda m_j = \mu_j$ for each $j$, so that $m_1=\ldots \ldots = m_{d-1}$
and $\frac{m_d}{m_1}=\left(1+\frac{\delta^2}{1-\delta^2}d\right)/2\ge 1$
for an integral solution.
In particular, $m_1= \frac{1}{\lambda}\left(\frac{1-\delta^2}{d}\right)$.
Since $\frac{1-\delta^2}{d}\le \frac{1}{d}$ and $m_i$ is a non-negative integer, 
it holds that $|\frac{1}{\lambda}|\ge d$.
Then, $|p|\ge d|C_\delta|$.
As a result, the protection offered by $p$ is much higher than that of
$n(\sigma_\delta)$, which is a contradiction.

\item $i\in \{1,\ldots,d-1\}$:
we rewrite 
$
C_\delta=\lambda \sum_{j=1}^{j=d} m_j v_\delta^j + (1-\lambda) v_\delta^i
= \lambda \sum_{j=1}^{j=d} n_j v_\delta^j,
$
where $n_i = m_i + 1/\lambda - 1$ and $n_j=m_j$ otherwise for $j\in [d]$.
Again, this gives us that
$\lambda n_j = \mu_j$ for each $j$, so that $n_1=\ldots \ldots = n_{d-1}$
and $\frac{n_d}{n_1}=\left(1+\frac{\delta^2}{1-\delta^2}d\right)/2\ge 1$.
We observe that  $m_1 = n_1 = n_i = m_i - 1+\frac{1}{\lambda}$, so that
$|m_i-m_1|=\left|\frac{1}{\lambda}-1\right|$.
This implies that $\mathrm{min}\{|m_i|,|m_1|\}\ge \frac{d-1}{2}$.
It is easy to verify that $p$ offers higher protection than 
$n(\sigma_\delta)$, which is a contradiction.
\end{itemize}

The last possibility is that the quantity 
$\left(1+\frac{\delta^2}{1-\delta^2}d\right)$
may be irrational, in which case all $m_i$s can not be integers, which
is also a contradiction to our assumption that $p$ is a lattice point.
\end{proof}

Finally, we show that
\begin{lemma}
\label{lemma:same-protection}
Each simplex of the Delaunay triangulation of 
$T_\delta(\Z^d)$ has the same protection for $\delta\in[0,1]$.
\end{lemma}

\begin{proof}
We prove the claim for each simplex that is formed on the distortions of
the vertices of the cube $[-1,1]^d$.
Every other simplex in the triangulation is a translation of one of
these simplices, so the claim follows.

By the monotone property of the Freudenthal triangulation, each simplex 
can be represented as a permutation $\pi$ of $[1,\ldots,d]$; 
to create the $i$-th element of the chain, $1$ is added to the value
at the $\{(d+1)-\pi(i)\}$-th co-ordinate of the $(i-1)$-th element.
For instance, the chain corresponding to the simplex $\sigma$ is 
$\{(0,\ldots,0),(0,\ldots,0,1),(0,\ldots,0,1,1),\ldots,(1,\ldots,1)\}$,
which is obtained
by adding 1 to the co-ordinate positions $\{d,d-1,d-2,\ldots,2,1\}$ in order,
so $\pi_\sigma$ is simply $(1,2,\ldots,d)$, the identity.
To calculate the protection for any other simplex $\tau_\delta$, we need to find 
the circumcenter $C_\delta (\tau_\delta)$ and the elements of $n(\tau_\delta)$.
Both of these are obtained by (indirectly) using the permutation $\pi_\tau$ 
on $C_\delta(\sigma_\delta)$ and on $n(\sigma_\delta)$.
It can be verified that $\tau_\delta$ has the same
protection as that of $\sigma_\delta$.
\end{proof}

We arrive at the results of Theorem~\ref{theorem:diagpro} by
normalizing the protection values of Equation~\eqref{equation:rawprotection}
with the radius 
$R_\delta=\sqrt{\frac{\delta^4(d^2-1)+\delta^2(d^2+2)+d^2-1}{12d}}$.

\subsection{Thickness and Aspect ratio}
\label{subsection:otherquality}

In this sub-section we prove the claims of Theorem~\ref{theorem:thickaspect}.
\paragraph{Height}
First, we calculate the heights of the distorted simplex.
Let $H_\delta^i$ denote the hyperplane passing through all vertices of 
$\sigma_\delta$ except $v_\delta^i$. 
It can be calculated that
\begin{equation}
H_\delta^i:
\begin{cases}
\{\delta-1,\ldots,\delta-1,-(1+(d-1)\delta)\} \cdot x + d\delta=0 & 
\text{for $i=0$}\\
(\underbrace{0,\ldots,0}_{d-i},-1,1,\underbrace{0,\ldots,0}_{i-2})\cdot x=0
& \text{for $i=1,\ldots,d-1$} \\
\{1+(d-1)\delta,1-\delta,\ldots,1-\delta\}\cdot x=0  & \text{for $i=d$}\\
\end{cases}.
\end{equation}
Let $h_\delta^i$ be the height of $v_\delta^i$ to $H_\delta^i$. It follows that 
$h_\delta^0=h_\delta^d=\frac{\delta\sqrt{d}}{(\delta^2d-\delta^2+1)^{1/2}}$
and $h_\delta^1=\ldots=h_\delta^{d-1}=1/\sqrt{2}$.
We can see that $h_0^0=0, h_{1/\sqrt{d+1}}^0=1/\sqrt{2}, h_1^0=1$ and 
$\frac{\partial h_\delta^0}{\partial \delta}=\frac{\sqrt{d}}{(\delta^2d-\delta^2+1)^{3/2}}>0$.
Hence, 
\begin{equation}
 min(h_\delta^i)=
 \begin{cases}
 \frac{1}{\sqrt{2}} & \text{for $1\ge\delta\ge \frac{1}{\sqrt{d+1}}$}\\
 \frac{\delta\sqrt{d}}{\sqrt{\delta^2d-\delta^2+1}} & 
 \text{for $\frac{1}{\sqrt{d+1}} \ge\delta > 0$} \\
 \end{cases}.
\end{equation}

\paragraph{Longest edge}
Now we calculate the longest edge of $\sigma_\delta$.
For any $0\le i<j \le d$, we see that 
$|v_\delta^j-v_\delta^i|=\frac{\sqrt{j-i}\sqrt{d-(1-\delta^2)(j-i)}}{\sqrt{d}}$.
Since $j-i$ can take $d$ distinct values, the edge lengths can take $d$ 
distinct values,
$l_x=\frac{\sqrt{dx-(1-\delta^2)x^2}}{\sqrt{d}}$ where $x:=j-i$ and 
$1\le x\le d$.
To find the longest length, we instead try to maximize the function
$f(x)=l_x^2d=dx-(1-\delta^2)x^2$ since both have the same maxima.
Then, $f(1)=d-(1-\delta^2),f(d)=\delta^2d^2$ and $f'(y)$ is 0 
at $y=\frac{d}{2(1-\delta^2)}$, where $f(y)=\frac{d^2}{4(1-\delta^2)}$. 
Checking for the maxima among these, it turns out that
\begin{equation}
 max(l_x)=
 \begin{cases}
  \delta\sqrt{d} & \text{for $1 \ge \delta\ge \frac{1}{\sqrt{2}}$}\\
  \frac{\sqrt{d}}{2\sqrt{1-\delta^2}} & \text{for $\frac{1}{\sqrt{2}} \ge \delta > 0$}
 \end{cases}.
\end{equation}

Thickness is defined as 
$\Theta_\delta=\frac{min_{0\le i\le d}(h_\delta^i)}{max(l_x)}$.
Aspect ratio is defined as 
$\Gamma_\delta=\frac{min_{0\le i\le d}(h_\delta^i)}{2R_\delta}$.
Using the value 
$R_\delta=\sqrt{\frac{3\delta^2 d^2 + (1-\delta^2)^2(d^2-1)}{12d}}$,
and substituting the rest of the quantities, we arrive at the results
of Theorem~\ref{theorem:thickaspect}.

\section{Acknowledgements}
Aruni Choudhary acknowledges support of ERC grant StG 757609.
Arijit Ghosh is supported by Ramanujan Fellowship (No. SB/S2/RJN-064/2015).

\newcommand{\etalchar}[1]{$^{#1}$}

\end{document}